\documentclass[sigconf,authorversion]{acmart}

\usepackage{tikz}
\usetikzlibrary{trees}
\usepackage{booktabs}
\usepackage{listings}
\usepackage{peg}

\usepackage{bcprules,bussproofs,multicol,synttree,stmaryrd}
\usepackage{amsmath,amssymb,amsthm,wasysym}

\newcommand{\lb}[1]{{\color{magenta}\mathbf{#1}}}
\newcommand{\valspace}[1]{~|_{#1}~}
\newcommand{\CAP}[2]{\{ {#2}\; \#\lb{#1} \}}
\newcommand{\LFOLD}[3]{{{#2}\textcolor{dkred}{\wedge\kern-.3em{^\ast}}\CAP{#1}{#3}}}
\newcommand{\EPEG}{\Xe}

\copyrightyear{2019} 
\acmYear{2019} 
\setcopyright{acmlicensed}
\acmConference[SAC '19]{The 34th ACM/SIGAPP Symposium on Applied Computing}{April 8--12, 2019}{Limassol, Cyprus}
\acmBooktitle{The 34th ACM/SIGAPP Symposium on Applied Computing (SAC '19), April 8--12, 2019, Limassol, Cyprus}
\acmPrice{15.00}
\acmDOI{10.1145/3297280.3297433}
\acmISBN{978-1-4503-5933-7/19/04}
\begin{document}

	%% Title information
\title[CPEG]{CPEG: A Typed Tree Construction from Parsing Expression Grammars with Regex-Like Captures}
%% when present, will be used in
%% header instead of Full Title.
%%\titlenote{}        %% \titlenote is optional;
%% can be repeated if necessary;
%% contents suppressed with 'anonymous'
%\subtitle{A Typed Tree Construction from Parsing Expression Grammars with Regex-Like Captures}                     %% \subtitle is optional
%\subtitlenote{with subtitle note}       %% \subtitlenote is optional;
%% can be repeated if necessary;
%% contents suppressed with 'anonymous'

%% Author information
%% Contents and number of authors suppressed with 'anonymous'.
%% Each author should be introduced by \author, followed by
%% \authornote (optional), \orcid (optional), \affiliation, and
%% \email.
%% An author may have multiple affiliations and/or emails; repeat the
%% appropriate command.
%% Many elements are not rendered, but should be provided for metadata
%% extraction tools.

%% Author with single affiliation.
\author{Daisuke Yamaguchi}
\affiliation{%
	\department{Graduate School of Engineering}
	\institution{Yokohama National University}
	\city{Yokohama} 
	\state{Kanagawa}
	\postcode{240-8501}
	\country{Japan}
}
\email{yamaguchi-daisuke-bf@ynu.jp}

\author{Kimio Kuramitsu}
\orcid{1234-5678-9012}
\affiliation{
	\department{Department of Mathematical and Physical Sciences}
	\institution{Japan Women's University}
	\city{Bunkyo-ku}
	\state{Tokyo}
	\postcode{112-868}
	\country{Japan}
}
\email{kuramitsuk@fc.jwu.ac.jp}

%% Abstract
\begin{abstract}
	CPEG is an extended parsing expression grammar with regex-like capture annotation. Two annotations (capture and left-folding) allow a flexible construction of syntax trees from arbitrary parsing patterns. 
	More importantly, CPEG is designed to guarantee structural constraints of syntax trees for any input strings. This reduces the amount of user code needed to check whether the intended elements exist.
	
	To represent the structural constraints, we focus on regular expression types, a variant formalism of tree automata, which have been intensively studied in the context of XML schemas. Regular expression type is inferred from a given CPEG by the type inference that is formally developed in this paper. We prove the soundness and the uniqueness of the type inference. The type inference enables a CPEG to serve both as a syntactic specification of the input and a schematic specification of the output.
\end{abstract}

%% 2012 ACM Computing Classification System (CSS) concepts
%% Generate at 'http://dl.acm.org/ccs/ccs.cfm'.
\begin{CCSXML}
	<ccs2012>
	<concept>
	<concept_id>10011007.10011006.10011039.10011040</concept_id>
	<concept_desc>Software and its engineering~Syntax</concept_desc>
	<concept_significance>500</concept_significance>
	</concept>
	<concept>
	<concept_id>10011007.10011006.10011041.10011688</concept_id>
	<concept_desc>Software and its engineering~Parsers</concept_desc>
	<concept_significance>500</concept_significance>
	</concept>
	</ccs2012>
\end{CCSXML}

\ccsdesc[500]{Software and its engineering~Syntax}
\ccsdesc[500]{Software and its engineering~Parsers}
%% End of generated code

%% Keywords
%% comma separated list
\keywords{Parsing, parsing expression grammars, type inference, regular expression types}	

%% \maketitle
%% Note: \maketitle command must come after title commands, author
%% commands, abstract environment, Computing Classification System
%% environment and commands, and keywords command.
\maketitle

\section{Introduction}
Regular expressions, or regexes, are a popular language tool that describe a complex text pattern and can match some set of strings.
Its popularity, however, results not only from excellent pattern matching but also declarative data extraction, in which any sub-pattern within parentheses will be captured as numbered groups. 
Since the capture capability provides a straightforward means for extracting parsed data, many small parsers have been implemented simply with the help of  regexes.

The aim of our study is to bring the regex-like capture capability into parsing expression grammars, or PEGs \cite{POPL04_PEG}. 
The reason for this aim is that PEGs are more expressive than regular expressions, which would enable us to extract some nested data (such as XML and JSON) that cannot be recognized by regexes. 
On the other hand, the capture capability would provide PEGs with a more straightforward means for implementing a parser than existing parser generators with action code. 

Nez parser \cite{Nez} is based on a declaratively extended parsing expression grammar with regex-like capture annotation. Two extended annotations (capture and left-folding) allow a flexible construction of syntax trees from arbitrary parsing patterns. We have demonstrated that Nez can parse many programming languages including Java, JavaScript, and Python.

While the declarative tree construction in Nez is convenient, it is still weak since the constructed trees are untyped. This means that they are treated as a common structure of tree data. To traverse their contents safely, the users need to check whether the intended elements exist.
Parser generators such as Yacc \cite{Yacc} and ANTLR \cite{PLDI11_Antlr} can produce typed trees throughout their programmed action code, which seems more suitable for handling complex syntax trees. Similarly, typed trees, or trees whose structures are well guaranteed, are desirable in the declarative tree construction.

The main challenge of this paper is to provide typing rules for syntax trees that will be captured by a PEG. A critical issue is to infer a type of tree from a grammar, before constructing concrete trees. 
This indicates that a declarative grammar guarantees some structural constraints of parse trees of any input strings. This property could be a good foundation for further static binding with a programming language. 

As the first attempt to infer a type of tree, we carefully designed CPEG to produce labeled unranked trees, which are equal to XML documents \cite{Hosoya:2010:FXP:1941946} (Note that CPEGs are a substantial subset of the Nez grammars). Here, we focus on \emph{regular expression types}, or RETs, which have been studied in a foundation of XML schemas and a type system of tree automata \cite{Hosoya:2003:XST:767193.767195}.
The type inference for CPEGs developed in this paper shows that CPEGs become a schematic specification of the output and not only a syntactic specification of the input.

The remainder of this paper proceeds as follows. 
In Section 2, we describe our motivation using Nez grammar and RETs.
In Section 3, we formally develop a CPEG as a string-to-tree transducer.
In Section \ref{formaltypingdef}, we introduce regular expression types and define type inference rules for CPEG.
In Section 5, we prove the soundness and the uniqueness of our type inference.
In Section \ref{relwork}, we review related work, and Section \ref{conc} concludes this paper.

\section{Motivating Examples}\label{BC}

We describe a motivation of tree construction and its typing in the context of PEGs. 
Here we use Nez grammar \cite{NezGrammar}, which is an open source implementation of CPEG. 

\subsection{Parsing Expression Grammars}

Nez grammar is a PEG-based grammar specification language, 
whose constructs come from those of PEGs.
Nez grammar is a set of syntax rules that are defined by a mapping from a non-terminal $A$ to a parsing expression $\EPEG$:
\[
A = \EPEG
\]

\begin{table}[bt]
	\centering
	\begin{tabular}{llll} \toprule
		PEG  & Type & Description\\ \midrule
		\verb|' '| & Primary  & Matches text\\
		$[ ]$ & Primary  & Matches character class \\
		$.$ & Primary  & Any character\\
		$A$ & Primary  & Nonterminal application\\
		$(\EPEG)$ & Primary  & Grouping\\
		$\EPEG\OPT$ & Unary suffix  & Option\\
		$\EPEG\ZERO$ & Unary suffix  & Zero-or-more repetitions\\
		$\EPEG\ONE$ & Unary suffix  & One-or-more repetitions\\
		$\AND\EPEG$ & Unary prefix  & And-predicate\\
		$\NOT\EPEG$ & Unary prefix  & Negation\\
		$\EPEG_1\CAT\EPEG_2$ & Binary  & Sequencing\\
		$\EPEG_1\ORE\EPEG_2$ & Binary  & Prioritized Choice\\ \bottomrule
	\end{tabular}
	\caption{PEG operators} 
	\label{table:peg}
\end{table}

Table \ref{table:peg} shows a list of PEG operators used in Nez which inherits the formal interpretation of PEGs \cite{POPL04_PEG}. This indicates that the string \verb|'abc'| exactly matches the same input, while \verb|[abc]| matches one of these characters. The . operator matches any single character. The $\EPEG\OPT$, $\EPEG\ZERO$, and $\EPEG\ONE$ expressions behave as in common regular expressions, except that they are greedy and match until the longest position. The $\EPEG_1\CAT\EPEG_2$ attempts two expressions $\EPEG_1$ and $\EPEG_2$ sequentially, backtracking to the starting position if either expression fails. The choice $\EPEG_1\ORE\EPEG_2$ first attempt $\EPEG_1$ and then attempt $\EPEG_2$ if $\EPEG_1$ fails. The expression $\AND\EPEG$ attempts $\EPEG$ without consuming any character. The expression $\NOT\EPEG$ succeeds if $\EPEG$ fails but fails if $\EPEG$ succeeds. 

In general, PEGs can express all languages that can be expressed by {\em deterministic} context-free grammars (such as LALR and LL($k$) grammars).  

PEGs, on the other hand, provide no specification for the output of a parser, 
while the parser users require certain forms of syntax trees that contain all necessary information for further processing. 
Note that we can regard non-terminals as a tree constructor in a way that a labeled tree node $A[...]$ is constructed from $A = \EPEG$. This approach is similar to that of derivation trees, resulting in redundantly nested trees. Besides, some forms of trees are not well constructed, as described in Section \ref{fold-capture}.

\subsection{Tree Annotation}

PEG only provides the syntactic matching capability while Nez provides two additional annotations (called \emph{capture} and \emph{fold-capture}) to construct complex syntax trees in a parser context. 

\subsubsection{Capture}\label{capanno}

The capture annotation $\CAP{L}{\Xe}$ is a straightforward extension to specify a parsing expression $\Xe$ (with enclosing braces $\{~~\}$) that extracts its matched string as a tree node. The extracted node is labeled by the given $\#\lb{L}$, which is used to identify the type of tree nodes. 

To start, let us consider a simple example, \pe{VAL} rule, which recognizes a sequence of numeric characters.
Here, \pe{Val} is a capture version of \pe{VAL} whose matched strings are constructed as a node labeled as \#Int.

Let us start a token extraction. 

\begin{peg}[morekeywords[0]={Int}]
	VAL = [0-9]\+
	Val = { [0-9]\+ #Int }
\end{peg}

We use $x \Parse{\Xe} v$ to write that an expression $\Xe$ parses an input string $x$ 
and then transforms it into a tree $v$. 

Here, we show how \pe{Val} parses the string \verb|123|: 

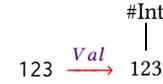
\begin{center}
	\branchheight{0.3in}
	\texttt{123} $\Parse{Val}$ \synttree{2}[\#Int[123]] 
\end{center}

The nested capture is naturally interpreted as a nested containment of trees that are constructed inside the braces. The following \pe{Prod2} parses from \texttt{123*45} to a tree depicted in the following manner:

\begin{peg}
	Prod2 = { Val '*' Val #Mul }
\end{peg}

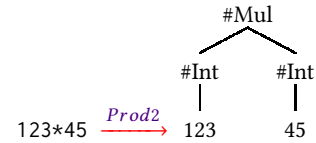
\begin{center}
	\branchheight{0.3in}
	\texttt{123*45} $\Parse{Prod2}$ \synttree{3}[\#Mul[\#Int[123]][\#Int[45]]]
\end{center}

The \pe{Prod2} accepts only a single multiplication expression. 
Multiple multiplications such as \verb|1*2*5| can be captured in two different forms.

\begin{peg}
	ProdM = { Val ('*' Val )\* #Mul }
	Prod = { Val ('*' Prod ) #Mul } \/ Val
\end{peg}

The \pe{ProdM} and \pe{Prod} both accept the same inputs, while they produce different forms of trees. 
The \pe{ProdM} uses a repetition which forms a variable-length list of trees. 
The \pe{Prod}, on the other hand, uses a recursion which forms a recursively nested tree. 

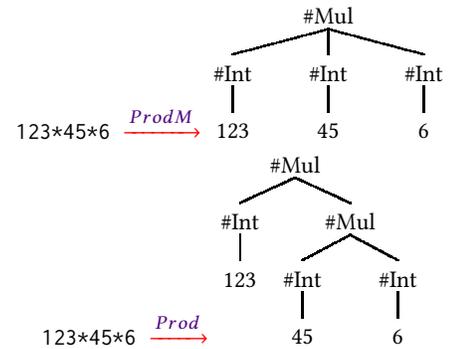
\begin{center}
	\branchheight{0.3in}
	\texttt{123*45*6} $\Parse{ProdM}$ 
	\synttree{3}[\#Mul
	[\#Int[123]]
	[\#Int[45]]
	[\#Int[6]]
	]
\end{center}

\begin{center}
	\branchheight{0.3in}
	\texttt{123*45*6} $\Parse{Prod}$ 
	\synttree{4}[\#Mul
	[\#Int[123]]
	[\#Mul
	[\#Int[45]]
	[\#Int[6]]
	]
	]
\end{center}

As shown, we can switch a form of trees by repetition and recursion. 

\subsubsection{Fold-Capture} \label{fold-capture}

In the previous section, we see that a recursion of capturing produces a nested form of recursive trees.
It is important to note that the produced trees are always in the right-associative form since PEGs do not allow left recursions. 

\begin{peg}
	ProdL = { (ProdL \/ Val) '*' Val #Mul }
\end{peg}

This example suggests that the capture does not well describe arbitrary forms of syntax trees, especially left-associative binary operators. 

CPEG additionally provides the fold-capture annotation, denoted as $\FOLD\{ \Xe\; \#\lb{L} \} $, which constructs a tree containing the left-hand tree as its first subtree.

Here is an example of a fold-capture version of \pe{Prod2}, where the first \pe{Val} is factored outside and then folded inside by $\FOLD$. 

\begin{peg}
	Prod2 = Val \^{ '*' Val #Mul } 
\end{peg}

The $\FOLD$ operator is left-associative. The trees on the left hand are folded one after another. 

\begin{peg}
	((Val \^{ '*' Val #Mul} ) ..) \^{ '*' Val #Mul}
\end{peg}

We use the repetition to denote this iterative folding. 
Finally, the left-associative version of \pe{Prod} is described as follows:

\begin{peg}
	ProdL = Val (\^{ '*' Val #Mul })\*
\end{peg}

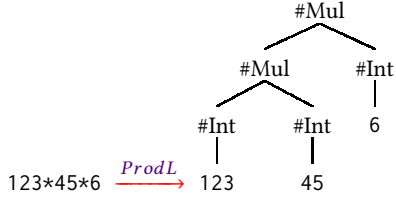
\begin{center}
	\branchheight{0.3in}
	\texttt{123*45*6} $\Parse{ProdL}$ 
	\synttree{4}[\#Mul
	[\#Mul
	[\#Int[\texttt{123}]]
	[\#Int[\texttt{45}]]
	]
	[\#Int[\texttt{6}]]
	]
\end{center}

In fact, left-recursion is a significant restriction of PEG. Although there is a known algorithm for eliminating any left-recursion from a grammar \cite{OOPSLA14_Antlr}, this algorithm does not ensure the left-associative. Using fold-capture annotation, we can replace the left-recursion under the impression of keeping the associativity.

\subsection{Typing with Regular Expression Types}\label{tret}
The declarative tree annotations in Nez are convenient and powerful enough to express many types of syntax trees, ranging from XML and JSON to Java and JavaScript \cite{Nez}. 
Trees that the parser users receive are untyped, which are the so-called \emph{common trees} that are formed in a common structure.

To traverse the trees safely, the users need to check whether the traversed tree is in an intended structure (i.e. its label and arity), as an XML schema variation.
Nevertheless, embedding the checking code in a traversal function makes the program fuzzier and prone to errors \cite{Petricek:2016:TDM:2908080.2908115}.

Besides, the traversal function should be implemented carefully so that the function is exhaustive\textemdash that is, some rule should be applied for all possible structured input trees.
Static analysis techniques for checking exhaustiveness checking are based on a type of input value. For instance, XDuce  \cite{Hosoya:2003:XST:767193.767195}, which is a statically typed XML processing language, provides an exhaustiveness checker that is found on a type called regular expression type \cite{Hosoya:2005}.
Regular expression types are a variant expression of tree automata \cite{tata2007}, which has been developed in the contexts of XML schema validations.

As our first attempt at the schematic variation on the trees and the statical exhaustiveness checking, we propose typing rules that infer a regular expression type for a given grammar.

Using regular expression types, all trees that can be parsed from \pe{ProdM} in the previous section have a type $\mathtt{ProdM}$, which is defined as follows:
\begin{align*}
\mathtt{type~ProdM}&=\lb{Mul}[\mathtt{Val},\mathtt{Val}^\ast]\\
\mathtt{type~Val}&=\lb{Int}[\mathtt{Empty}]
\end{align*}

Here, the type $\lb{Mul}[\mathtt{Val},\mathtt{Val}^\ast]$ describes a tree with $\lb{Mul}$ that has subtrees typed by the elements of \verb|[   ]| . The `$,$', `$\ast$', and $\mathtt{Empty}$ inside respectively denotes \emph{concatenation}, \emph{repetition}, and \emph{empty tree} as regular expression operators. This indicates that the $\lb{Mul}$ tree has one or more subtrees that are typed with $\mathtt{Val}$, which is a type variable defined in the second line. The type $\lb{Int}[\mathtt{Empty}]$ represents no more subtrees and it's type is Int. 

Regular expression types allow recursive type definition.
All trees that can be parsed from \pe{ProdL} are described in the following type:
\begin{align*}
\mathtt{type~ProdL}&=\lb{Prod}[\mathtt{ProdL},\mathtt{Val}]~|~\mathtt{Val}\\
\mathtt{type~Val}&=\lb{Int}[\mathtt{Empty}]
\end{align*}

\section{Formal Definition of CPEG}\label{formaldef}

CPEG is a subset of Nez grammar \cite{Nez} to highlight a focused extension to the construction of labeled unranked trees.

\subsection{Grammar}

A CPEG $G$ is a 5-tuple $G=(N_G,\Sigma,P_G,\Xe_s,\mathcal{S})$, where $N_G$ is a finite set of non-terminal symbols, $\Sigma$ is a finite set of terminal symbols, $P_G$ is a finite set of \textit{production rules}, $\Xe_s$ is a \textit{start expression} and $\mathcal{S}$ is a finite set of label symbols.

Each production rule $r\in P_G$ is a pair $(A,\Xe)$, that is written as $A\leftarrow \Xe$, where $A\in N_G$ and $\Xe$ is an expression. For any $A\in N_G$, there is exactly one $\Xe$ such that $A\leftarrow \Xe$. We regard $P_G$ as a function from non-terminals to expressions.

The syntax of an expression $\Xe$ is defined in Figure \ref{op}.

\begin{figure}[ht]
	\begin{equation*}
	\begin{array}{lllll}
	e& ::=&\varepsilon& & \mbox{empty} \\
	& & a&(a\in\Sigma)& \mbox{terminal} \\
	& & A&(A\in N_G)& \mbox{nonterminal} \\
	& & \Xe_1\CAT\Xe_2& &\mbox{sequence} \\
	& & \Xe_1\ORE \Xe_2& &\mbox{ordered choice} \\
	& & \Xe\ZERO&& \mbox{repetition} \\
	& & \NOT\Xe&& \mbox{not-predicate} \\
	& & \CAP{L}{\Xe}&(\lb{L}\in\mathcal{S})&\mbox{capture}\\
	& & \LFOLD{L}{\Xe_1}{\Xe_2} & (\lb{L}\in\mathcal{S})&\mbox{fold-capture}\\
	\end{array}
	\end{equation*}
	\caption{Expressions in CPEG}
	\label{op}
\end{figure}

The syntax of CPEGs is reasonably simple. Due to the following syntactic sugars (also defined in \cite{POPL04_PEG}), CPEG expresses all the parsing expressions that are constituted by PEG operators (also in Table \ref{table:peg}).

\begin{eqnarray}
'abc' & = & 'a' \ 'b' \ 'c' \notag \\
\left[abc\right]  & = & 'a' \ORE \ 'b' \ORE 'c' \notag \\
\Xe\ONE & = & \Xe\CAT\Xe\ZERO \notag \\
\Xe\OPT & = & \Xe \ORE \ \varepsilon \notag \\
\AND \Xe & = &  \NOT(\NOT \Xe) \notag
\end{eqnarray}

Similarly, the fold-capture annotation can be rewritten by the capture annotation, such as:
\begin{eqnarray}
\Xe_1 \wedge \CAP{L}{\Xe_2} & = &  \CAP{L}{\Xe_1\CAT\Xe_2}  \notag \\
\Xe_1 (\wedge \CAP{L}{\Xe_2})\ORE (\wedge \CAP{L'}{\Xe_3})& = &  \CAP{L}{\Xe_1\CAT\Xe_2} \ORE \CAP{L'}{\Xe_1\CAT\Xe_3}  \notag \\
\Xe_1 (\wedge \CAP{L}{\Xe_2})\OPT & = &  \CAP{L}{\Xe_1\CAT\Xe_2} \ORE \Xe_1 \notag
\end{eqnarray}

An important exception is the repetitive combination such as $\Xe_1 (\wedge \CAP{L}{\Xe_2})\ZERO$,
which can not be rewritten by the capture annotation.
Thus, we focus only on the repetitive fold-capture as defined in Figure \ref{op}.
Note that $\Xe_1 (\wedge \CAP{L}{\Xe_2})\ZERO$ is a syntax sugar of $\CAP{L}{A \CAT\Xe_2}$ where $A = \CAP{L}{A\CAT\Xe_2} / \Xe_1$. However, this syntactic sugar contains a left recursion that never terminates in PEGs.

\paragraph{Well-Formed Grammars}
As will be discussed in Section \ref{formaltypingdef}, a regular expression type is inferred for each CPEG.
Some inferred types, however, do not hold a \emph{well-formedness} condition \cite{Hosoya:2005} without any syntactic restriction on CPEG.
Here, we impose an additional restriction to the syntax of CPEG. We refer to the CPEG that holds the restriction as \emph{well-formed grammar}.

Well-formed grammar is a CPEG that holds either of the following two conditions:
\begin{itemize}
	\item recursive use of nonterminals occurs only in tail positions.
	\item if a sequence $A~\Xe$ exists in the production rules, where $A$ is a recursively used nonterminal and $\Xe$ is a some expression (i.e., $A$ is not in tail positions); then $\Xe$ has neither capture nor fold-capture as a subexpression.
\end{itemize}

For instance, the CPEG that has a production rule as the following:
\[A\leftarrow \CAP{L_1}{\Xe_{1}}A\CAP{L_2}{\Xe_{2}}~/~\Xe\]
is not a well-formed grammar, because $A$ is not in the tail position and the subsequent expression of $A$ is a fold-capture.

Whereas, the CPEG that has a production rule as follows:
\[A\leftarrow \CAP{L_1}{\Xe_{1}}A~/~\Xe\]
is a well-formed grammar, because $A$ is in the tail position.

Furthermore, the CPEG which have a production rule as follows:
\[A\leftarrow \CAP{L_1}{\Xe_{1}}Aa~/~\Xe\]
is a well-formed grammar while $A$ is not in the tail position, because the subsequent expression of $A$ is a terminal $a$.

Naturally, the CPEG which have production rules:
\begin{align}
A&\leftarrow \CAP{L_1}{\Xe_{1}}Ba~/~\Xe \notag\\
B&\leftarrow \CAP{L_2}{\Xe_{2}}A~/~\Xe' \notag
\end{align}
is a well-formed grammar as well.

\subsection{Formal Interpretation of CPEG}
The semantics is defined as a relation from expressions to trees.

\subsubsection{Tree}

We start by defining a textual notation to denote a tree for convenience.

The tree of CPEG is a labeled unranked tree on $\mathcal{S}$ and $\Sigma^\ast$,
where the nodes are labeled with $\lb{L}\in\mathcal{S}$ and the leaves hold a string $x\in\Sigma^\ast$.

Let $\mathcal{T}_{\mathcal{S}\times\Sigma^\ast}$ be a set of trees on $\mathcal{S}$ and $\Sigma^\ast$.
The syntax of a tree $v\in\mathcal{T}_{\mathcal{S}\times\Sigma^\ast}$ is defined as follows.
\begin{equation*}
\begin{array}{lllll}
v& ::=&  \lb{L}[v]\quad&(\lb{L}\in\mathcal{S})~&\mbox{node} \\
& & v_1,v_2 &&\mbox{concatenation} \\
& &x& (x\in\Sigma^\ast)&\mbox{string}  \\
\end{array}
\end{equation*}

Node $\lb{L}[v]$ denotes the node that is labeled with a symbol $\lb{L}$ and has subtrees denoted by $v$.
We use a concatenation operator \verb|,| to handle multiple subtrees.
Since the order of subtrees is preserved, the concatenation \verb|,| is not commutative.
The notation $x$ denotes a string including empty string.
We assume that a concatenation of two strings is equal to a single string. That is $x,x'$ is equal to $xx'$.
Additionally, we assume a concatenation of a node and a string is equal to the node. That is $\lb{L}[v],x$ and $x,\lb{L}[v]$ are equal to $\lb{L}[v]$.
This premise reads that if a tree $\lb{L}[v],x$ or $x,\lb{L}[v]$ is constructed, the tree can be regarded as $\lb{L}[v]$.
In other words, a string that is concatenated with a node can be ignored.

Here, we show the same trees with both pictorial and textual notations as follows.\\
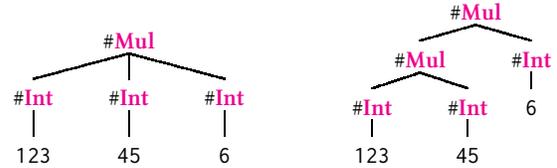
\begin{minipage}[b]{.44\linewidth}
	\centering
	\branchheight{0.3in}
	\synttree{3}[\#$\lb{Mul}$
	[\#$\lb{Int}$[\texttt{123}]]
	[\#$\lb{Int}$[\texttt{45}]]
	[\#$\lb{Int}$[\texttt{6}]]
	]
	\[\lb{Mul}[\lb{Int}[\mathtt{123}],\lb{Int}[\mathtt{45}],\lb{Int}[\mathtt{6}]]\]
\end{minipage}
\begin{minipage}[b]{.56\linewidth}
	\centering
	\branchheight{0.25in}
	\synttree{4}[\#$\lb{Mul}$
	[\#$\lb{Mul}$
	[\#$\lb{Int}$[\texttt{123}]]
	[\#$\lb{Int}$[\texttt{45}]]
	]
	[\#$\lb{Int}$[\texttt{6}]]
	]
	\[\lb{Mul}[\lb{Mul}[\lb{Int}[\mathtt{123}],\lb{Int}[\mathtt{45}]],\lb{Int}[\mathtt{6}]]\]
\end{minipage}

\subsubsection{Operational Semantics}\label{peg2tree}
To formalize the semantics of a grammar $G$,
we define a relation $\Downarrow_G$ from $(\Xe,x)$ to $(o,y)$,
where $\Xe$ is an expression of CPEG, $x$ is an input string, $y$ is an unconsumed string, and $o \in \mathcal{T}_{\mathcal{S}\times\Sigma^\ast}\cup\{f\}$ is an output.
If $o$ is in the form of a tree $v$, the matching succeeds and the tree $v$ is constructed from an input string $x$.
If $o=f$, the distinguished symbol $f$ indicates \emph{failure}.

For $((\Xe,x),(o,y)) \in \Downarrow_G$, we write $\Xe\Downarrow_y^x o$. This is read as ``\emph{an expression $\Xe$ parses an input $x$ and transforms it to an output $o$ with an unconsumed string $y$}''.

Now supposing $ v \in\mathcal{T}_{\mathcal{S}\times\Sigma^\ast}$, $a,b,c\in\Sigma$, $~x,y,z\in\Sigma^\ast$, and $\varepsilon$ is an empty string, $\Downarrow^x_y$ is the smallest relation closed under the set of rules shown in Figure \ref{down}.

\begin{figure*}[ht]
	\begin{multicols}{3}
		\infax[E-Empty]{\varepsilon\Downarrow_x^x\varepsilon}
		
		\infax[E-Term1]{a\Downarrow_x^{ax}a}
		
		\infax[E-Term2]{a\Downarrow_x^{bx}f\quad(a\ne b)}
		
		\infrule[E-Nt]{P_G(A)\Downarrow^x_yv}{
			A\Downarrow^x_y v}
		
		\infrule[E-Seq1]{\Xe_{1}\Downarrow^{x_1x_2y}_{x_2y}v_1\andalso \Xe_2\Downarrow^{x_2y}_yv_2}{
			\Xe_1\Xe_2\Downarrow^{x_1x_2y}_yv_1,v_2}
		
		\infrule[E-Seq2]{\Xe_{1}\Downarrow^x_xf}{
			\Xe_1\Xe_2\Downarrow^x_xf}
		
		\infrule[E-Seq3]{\Xe_{1}\Downarrow^{x_1y}_{y}v_1\andalso \Xe_2\Downarrow^y_yf}{
			\Xe_1\Xe_2\Downarrow^{x_1y}_{x_1y}f}
		
		\infrule[E-Alt1]{\Xe_1\Downarrow^{xy}_yv_1}{
			\Xe_1/\Xe_2\Downarrow^{xy}_yv_1}
		
		\infrule[E-Alt2]{\Xe_1\Downarrow^{xy}_{xy}f\andalso \Xe_2\Downarrow^{xy}_yv_2}{
			\Xe_1/\Xe_2\Downarrow^{xy}_yv_2}
		
		\infrule[E-Alt3]{\Xe_1\Downarrow^x_xf\andalso \Xe_2\Downarrow^x_xf}{
			\Xe_1/\Xe_2\Downarrow^x_xf}
		
		\infrule[E-Rep1]{\Xe\Downarrow^{x_1x_2y}_{x_2y}v_1\andalso \Xe\ZERO\Downarrow^{x_2y}_y v_2}{
			\Xe\ZERO\Downarrow^{x_1x_2y}_yv_1,v_2}
		
		\infrule[E-Rep2]{\Xe\Downarrow^x_xf}{
			\Xe\ZERO\Downarrow^x_x\varepsilon}
		
		\infrule[E-Not1]{\Xe\Downarrow^{xy}_yv}{
			\NOT\Xe\Downarrow^{xy}_{xy}f}
		
		\infrule[E-Not2]{\Xe\Downarrow^{x}_xf}{
			\NOT\Xe\Downarrow^x_x\varepsilon}
		\infrule[E-Capture1]{\Xe\Downarrow^{xy}_y v}{
			\CAP{L}{\Xe}\Downarrow^{xy}_y \lb{L}[v]}
		\infrule[E-Capture2]{\Xe\Downarrow^{x}_xf}{
			\CAP{L}{\Xe}\Downarrow^{x}_x f}
		
	\end{multicols}
	\infrule[E-FoldCap1]{\Xe_1\Downarrow^{xy_1y_2\dots y_nz}_{y_1y_2\dots y_nz} v_1\andalso \Xe_2\Downarrow^{y_1y_2\dots y_nz}_{y_2\dots y_nz}v_2
		\quad\cdots\quad \Xe_2\Downarrow^{y_nz}_{z}v_n\andalso \Xe_2\Downarrow^z_z f}{
		\LFOLD{L}{\Xe_1}{\Xe_2}\Downarrow^{xy_1y_2\dots y_nz}_{z}\lb{L}[\lb{L}[\dots\lb{L}[\lb{L}[v_1,v_2],v_3],\dots,v_{n-1}],v_n]}
	\begin{multicols}{2}
		\infrule[E-FoldCap2]{\Xe_1\Downarrow^{x}_{x}f}{
			\LFOLD{L}{\Xe_1}{\Xe_2}\Downarrow^{x}_x f}
		\infrule[E-FoldCap3]{\Xe_1\Downarrow^{xy}_y v_1\andalso \Xe_2\Downarrow^{y}_{y} f}{
			\LFOLD{L}{\Xe_1}{\Xe_2}\Downarrow^{xy}_{y} v_1}
	\end{multicols}
	\caption{Rules for the relation $\Downarrow_y^x$}
	\label{down}
\end{figure*}

The matching interpretation in CPEG is the same as PEG \cite{POPL04_PEG}, except for some tree construction.
The capture and fold-capture are an explicit tree constructor as in (\textsc{E-Capture1}) and (\textsc{E-FoldCap1}). All the trees constructed in the subexpressions are contained in a newly constructed tree.

The rules for fold-capture $\LFOLD{L}{\Xe_1}{\Xe_2}$ are (\textsc{E-FoldCap1}), (\textsc{E-FoldCap2}) and (\textsc{E-FoldCap3}).
By (\textsc{E-FoldCap1}), a left-associative tree is derived from $\LFOLD{L}{\Xe_1}{\Xe_2}$, if the presuppositions of the rule are satisfied. The presuppositions are summarized in two conditions.
The first condition is that the subexpression $\Xe_1$ derives a tree $v_1$ from an input string $xy_1y_2\dots y_nz$.
Now we denote the unconsumed string by $y_1y_2\dots y_nz$.
The second condition is that subexpression $\Xe_2$ derives a tree $v_2$ from the string $y_1y_2\dots y_nz$.
Additionally, if tree $v_3$, tree $v_4$, $\cdots$ and tree $v_n$ are derived, by applying $\Xe_2$ repeatedly to the unconsumed string until the out put goes failure, then the derived tree of $\LFOLD{L}{\Xe_1}{\Xe_2}$ is a left-associative tree:
\[\lb{L}[\underbrace{\lb{L}[\dots\lb{L}[\lb{L}[v_1}_{n-1},v_2],v_3],\dots,v_{n-1}],v_n].\]
The tree $v_1$ is contained at the left branch that is $n-1$ times nested from the root. The other tree $v_i$ ($2\leq i\leq n$) is stored at the right branch that is $n-i+1$ times nested from the root.
The derived tree is also described pictorially as below.

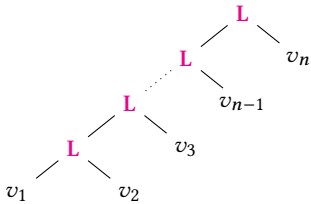
\begin{center}
	\begin{tikzpicture}[level distance=0.6cm]
	\node{$\lb{L}$}
	child {node{$\lb{L}$}
		child [dotted] {
			node{$\lb{L}$}
			child [solid] {node{$\lb{L}$}
				child {node{$v_1$}}
				child {node{$v_2$}}
			}
			child [solid] {node{$v_3$}}
		}
		child {node{$v_{n-1}$}}
	}
	child {node{$v_n$}};
	\end{tikzpicture}
\end{center}
By (\textsc{E-FoldCap2}), if the subexpression $\Xe_1$ goes to failure, $\LFOLD{L}{\Xe_1}{\Xe_2}$ also goes to failure.
By (\textsc{E-FoldCap3}), if the subexpression $\Xe_1$ derives $v_1$ but the subexpression $\Xe_2$ goes to failure, $\LFOLD{L}{\Xe_1}{\Xe_2}$ derives just $v_1$. Note that $v_1$ is not nested in the node with $\lb{L}$.

\section{Typing CPEG} \label{formaltypingdef}

\subsection{Regular Expression Types}
In this section, we review the syntax and the semantics of regular expression types (RETs) \cite{Hosoya:2005}.

A RET $\mathtt{T}$ is inductively defined as follows:
\begin{equation}
\begin{array}{llll}
{\mathtt{T}} &::=& \mathtt{Empty} &\mbox{empty sequence}\\
&&     \mathtt{T_1,T_2} &\mbox{concatenation}\\
&&     \mathtt{T_1}~|~\mathtt{T_2}&\mbox{union} \\
&&     \mathtt{T^*}&\mbox{repetition} \\
&&     \mathtt{\lb{L}[T_1,\dots,T_n]} & \mbox{label}\\
&&     \mathtt{X} & \mbox{type variable}\\
\end{array}
\notag
\end{equation}

Here, the semantics of RETs is given by the relation $v:\mathtt{T}$ ($v\in\mathcal{T}_{\mathcal{S}\times\Sigma^\ast}$), read ``the tree $v$ has type $\mathtt{T}$'' \textemdash the smallest relation closed under the set of typing rules in Figure \ref{subtyping}.

\begin{figure}
	\begin{multicols}{2}
		\infax[S-Empty]{x:\mathtt{Empty}}
		\infrule[S-Seq]{v_1:\mathtt{T}_1\andalso v_2:\mathtt{T}_2}{v_1,v_2:\mathtt{T_1,T_2}}
		\infrule[S-Or1]{v:\mathtt{T_1}}{v:\mathtt{T_1|T_2}}
		\infrule[S-Or2]{v:\mathtt{T_2}}{v:\mathtt{T_1|T_2}}
		\infrule[S-Rep]{v_i: \mathtt{T}~\mbox{for each}~i}{v_1,\dots,v_n:\mathtt{T^*}}
		\infrule[S-Node]{v:\mathtt{T}}{\lb{L}[v]:\lb{L}\mathtt{[T]}}
		\infrule[S-Var]{E(\mathtt{X})=\mathtt{T}\andalso v:\mathtt{T}}{v:\mathtt{X}}
	\end{multicols}
	\caption{Typing rules for trees}
	\label{subtyping}
\end{figure}

The empty sequence $\mathtt{Empty}$, the concatenation $\mathtt{T_1,T_2}$, the union $\mathtt{T_1}|\mathtt{T_2}$, and the repetition $ \mathtt{T^*}$ come from regular expressions, as its name implies.
The label $\lb{L}[\mathtt{T}]$ represents a tree that contains a subtree that has a type $\mathtt{T}$.

The $\mathtt{X}$ represents a type variable that binds an arbitrary RET.
The bindings of type variables are given by a single global set $E$ of type definitions of the following form:
\[\mathtt{type~X} = \mathtt{T}\]
We regard $E$ as a mapping from type variables to their bodies and write $E(\mathtt{X})$ for a reference to the mapped type from $\texttt{X}$ in $E$.

As shown above, the type variables allow recursive types.
The readers should note that the expressiveness of RETs correspond to some of the context-free grammars, despite the fact that its name is regular.

\subsection{Subtyping}
The subtype relation between two types is defined semantically: two types $\mathtt{S},\mathtt{T}$ are in the subtype relation $<:$ if and only if $v:\mathtt{S}$ implies $v:\mathtt{T}$
\[\mathtt{S}<:\mathtt{T}\iff \forall v\in\mathcal{T}_{\mathcal{S}\times\Sigma^\ast}.~v:\mathtt{S}\Rightarrow v:\mathtt{T}\]

As Hosoya \cite{Hosoya:2005}　indicated, the semantic notion of subtype relation over RETs immediately corresponds to the notion of inclusion relation on the set theory.

\subsection{Type Inference for CPEG}\label{typeinf}

CPEG intends to infer the types of trees from a grammar before constructing concrete trees.

To start, we consider that a tree $v$ is derived from an input $x$ with a CPEG $G=(N_G,\Sigma,P_G,\Xe_s,\mathcal{S})$.
Since the tree is constructed by the derivation $\Xe_s\Downarrow^x_y v$, the type of $G$ is regarded as the type of the start expression $\Xe_s$.

Now, we define type inference rules for CPEG expressions.
Let $E$ be a single global set of type bindings and $\Gamma$ be a type environment mapping from non-terminals to type variables. The mappings in $\Gamma$ are denoted by $A:\mathtt{X_A}$.

The inference rules are defined as a typing relation denoted by $\Gamma\vdash \Xe:\mathtt{T}\valspace{\chi}E$, which can be read  ``\emph{under a typing environment $\Gamma$, an expression $\Xe$ has a type \texttt{T} with a global set $E$}''.
The relations are the smallest relation closed under the set of typing rules shown in Figure \ref{typing}.

The $\chi$ is a set of type variables. The $\chi$ is used to store the type variables introduced in each subderivation, and $\mathtt{T}\valspace{\chi} E$ ensure that the variables appearing in $\mathtt{T}$ are fresh, for each condition of $\chi$ prevent us from building a derivation in which the same variable is used as ``fresh'' in two different places.
Since there is an infinite supply of type variable names, we can always find a way to satisfying the condition.
These conventions come from Pierce's textbook \cite{pierce2002types}.

\begin{figure}[ht]
	\infax[T-Empty]{\vdash\varepsilon:\mathtt{Empty}\valspace{\emptyset}\emptyset}
	\infax[T-Term]{\vdash a:\mathtt{Empty}\valspace{\emptyset}\emptyset}
	\infrule[T-Nt1]{A:\mathtt{X_A}\not\in \Gamma\andalso\Gamma,A:\mathtt{X_A}\vdash P_G(A):\mathtt{T}\valspace{\chi}E\andalso\{\mathtt{X_A}\}\cap\chi=\emptyset}{\Gamma\vdash A:\mathtt{X_A}\valspace{\chi\cup\{\mathtt{X_A}\}}E\cup\{\mathtt{type~X_A}=\mathtt{T}\}}
	\infrule[T-Nt2]{A:\mathtt{X}\in\Gamma}{\Gamma\vdash A:\mathtt{X}\valspace{\emptyset}\emptyset}
	\infrule[T-Seq]{\Gamma\vdash \Xe_1:\mathtt{T}_1\valspace{\chi_1}E_1\andalso \Gamma\vdash \Xe_2:\mathtt{T}_2\valspace{\chi_2}E_2\andalso \chi_1\cap\chi_2=\emptyset}{\Gamma\vdash \Xe_1\Xe_2:\mathtt{T_1,T_2}\valspace{\chi_1\cup\chi_2}E_1\cup E_2}
	\infrule[T-Alt]{\Gamma\vdash \Xe_1:\mathtt{T}_1\valspace{\chi_1}E_1\andalso \Gamma\vdash \Xe_2:\mathtt{T}_2\valspace{\chi_2}E_2\andalso \chi_1\cap\chi_2=\emptyset}{\Gamma\vdash \Xe_1/\Xe_2:\mathtt{T_1|T_2}\valspace{\chi_1\cup\chi_2}E_1\cup E_2}
	\infrule[T-Rep]{\Gamma\vdash \Xe:\mathtt{T}\valspace{\chi}E}{\Gamma\vdash \Xe\ZERO:\mathtt{T^\ast}\valspace{\chi}E}
	\infax[T-Not]{\vdash\NOT\Xe:\mathtt{Empty}\valspace{\emptyset}\emptyset}
	\infrule[T-Capture]{\Gamma\vdash \Xe:\mathtt{T}\valspace{\chi}E}{\Gamma\vdash\CAP{L}{\Xe}:\lb{L}[\mathtt{T}]\valspace{\chi}E}
	\infrule[T-FoldCap]{\Gamma\vdash \Xe_1:\mathtt{T_1}\valspace{\chi_1}E_1\andalso \Gamma\vdash \Xe_2:\mathtt{T_2}\valspace{\chi_2}E_2\\
		\chi_1\cap\chi_2=\emptyset\andalso \chi_1\cap\{\mathtt{X}\}=\emptyset\andalso \{\mathtt{X}\}\cap\chi_2=\emptyset}{\Gamma\vdash \LFOLD{L}{\Xe_1}{\Xe_2}:\mathtt{X}\valspace{\chi_1\cup\chi_2\cup\{\mathtt{X}\}}E_1\cup E_2\cup\{\mathtt{type~X}=\lb{L}[\mathtt{X},\mathtt{T_2}]|\mathtt{T_1}\}}
	\caption{Typing rules associated with a single global set $E$}
	\label{typing}
\end{figure}

The rules (\textsc{T-Nt1}) and (\textsc{T-Nt2}) are types for nonterminal symbols. The premise $A:\mathtt{T}\not\in\Gamma$ in \textsc{(T-Nt1)} is simply an explicit reminder: first, we check if the relation $A:\mathtt{T}$ is included in $\Gamma$; if not, we add that relation to $\Gamma$ and then try to calculate $\mathtt{T}$ for $P_G(A)$.
In the rule (\textsc{T-Nt1}), the sets $E$ and $\chi$ are updated to $E\cup\{\mathtt{type~X_A}=\mathtt{T}\}$ and $\chi\cup\{\mathtt{X_A}\}$ respectively.
The rules (\textsc{T-Seq}), (\textsc{T-Alt}), and (\textsc{T-FoldCap}) update the global set and $\chi$ as well.

Now, let us consider the typing of the following CPEG:
\begin{align*}
G=&(\{Prod,Val\},\{0,1,2,\cdots,9,\APLstar\},P_G,Prod,\{\lb{Prod},\lb{Int}\})\\
P_G=&\{Prod\leftarrow \LFOLD{Prod}{Val}{~\APLstar~ Val}, \\
&Val\leftarrow\CAP{Int}{[0-9]}\}
\end{align*}
Note that $[0-9]$ is a derived form of $0/1/2/3/\cdots/7/8/9$.

\begin{figure*}[ht]
	
	\begin{prooftree}
		\AxiomC{$\vdots$}
		\RightLabel{\textsc{(T-Capture)}}
		\UnaryInfC{$Prod:\mathtt{X_1},Val:\mathtt{X_3}\vdash\CAP{Int}{[0-9]}:\lb{Int}[\mathtt{Empty}]\valspace{\emptyset}\emptyset$}
		\RightLabel{\textsc{(T-Nt1)}}
		\UnaryInfC{$Prod:\mathtt{X_1}\vdash Val:\mathtt{X_3}\valspace{\{\mathtt{X_3}\}} E_2$}
		\AxiomC{\large$\mathcal{D}$}
		\RightLabel{\textsc{(T-FoldCap)}}
		\BinaryInfC{$Prod:\mathtt{X_1}\vdash\LFOLD{Prod}{Val}{\APLstar Val}:\mathtt{X_2}\valspace{\{\mathtt{X_2,X_3,X_4}\}}E_1$}
		\RightLabel{\textsc{(T-Nt1)}}
		\UnaryInfC{$\vdash Prod:\mathtt{X_1}\valspace{\{\mathtt{X_1,X_2,X_3,X_4}\}}E$}
	\end{prooftree}
	\begin{flushleft}
		Where,
		{\large$\mathcal{D}$}$=$
	\end{flushleft}
	\begin{prooftree}
		\AxiomC{}
		\RightLabel{\textsc{(T-Term)}}
		\UnaryInfC{$Prod:\mathtt{X_1}\vdash\APLstar:\mathtt{Empty}\valspace{\emptyset}\emptyset$}
		\AxiomC{$\vdots$}
		\RightLabel{\textsc{(T-Capture)}}
		\UnaryInfC{$Prod:\mathtt{X_1},Val:\mathtt{X_4}\vdash\CAP{Int}{[0-9]}:\lb{Int}[\mathtt{Empty}]\valspace{\emptyset}\emptyset$}
		\RightLabel{\textsc{(T-Nt1)}}
		\UnaryInfC{$Prod:\mathtt{X_1}\vdash Val:\mathtt{X_4}\valspace{\{\mathtt{X_4}\}}E_3$}
		\RightLabel{\textsc{(T-Seq)}}
		\BinaryInfC{$Prod:\mathtt{X_1}\vdash\APLstar Val:\mathtt{Empty},\mathtt{X_4}\valspace{\{\mathtt{X_4}\}}E_3$}
	\end{prooftree}
	
	\begin{align*}
	{E\;}&=\{\mathtt{type~X_1}=\mathtt{X_2}, \mathtt{type~X_2}=\lb{Prod}[\mathtt{X_2},~\mathtt{Empty},\mathtt{X_4}]~|~\mathtt{X_3}
	,\mathtt{type~X_3}=\lb{Int}[\mathtt{Empty}]
	,\mathtt{type~X_4}=\lb{Int}[\mathtt{Empty}]\}\\
	{E_1}&=\{\mathtt{type~X_2}=\lb{Prod}[\mathtt{X_2},~\mathtt{Empty},\mathtt{X_4}]~|~\mathtt{X_3}
	,\mathtt{type~X_3}=\lb{Int}[\mathtt{Empty}]
	,\mathtt{type~X_4}=\lb{Int}[\mathtt{Empty}]\}\\
	{E_2}&=\{\mathtt{type~X_3}=\lb{Int}[\mathtt{Empty}]\}\\
	{E_3}&=\{\mathtt{type~X_4}=\lb{Int}[\mathtt{Empty}]\}
	\end{align*}
	\caption{A derivation tree that derives $\vdash Prod:\mathtt{X_1}\valspace{\{\mathtt{X_1,X_2,X_3,X_4}\}}E$. This tree is omitted the freshness condition of type variables for simplicity.}
	\label{prooftree}
\end{figure*}

The derivation of typing is shown in Figure \ref{prooftree}.

First, $\CAP{Int}{[0-9]}$ has the type $\lb{Int}[\mathtt{Empty}]$ by \textsc{(T-Capture)}, where both $E$ and $\chi$ are empty.
Although $[0-9]$ in this premise has some type, we omit this derivation for readability.
We derive that $Val$ has a type variable $\mathtt{X_3}$ by \textsc{(T-Nt1)}.
At this derivation, the empty global set is updated to $E_2$ as shown in the bottom of Figure \ref{prooftree}.
Moreover, the empty $\chi$ is updated to $\{\mathtt{X_3}\}$. Note that $\mathtt{X_3}$ is newly introduced in this step.

The next derivation $\APLstar Val$ is a little complex. Hence, we make a sub-derivation tree $\mathcal{D}$. The sub-derivation proceeds similarly from the upper left of the tree to the bottom.
As a result of $\mathcal{D}$, the following is derived. \[Prod:\mathtt{X_1}\vdash~\APLstar~Val:\mathtt{Empty},\mathtt{X_4}\valspace{\{\mathtt{X_4}\}}E_3\]
where, ${E_3}=\{\mathtt{type~X_4}=\lb{Int}[\mathtt{Empty}]\}$.

Now, we can go back to the main derivation tree. The following is derived by \textsc{(T-FoldCap)}.
\[Prod:\mathtt{X_1}\vdash\LFOLD{Prod}{Val}{~\APLstar~ Val}:\mathtt{X_2}\valspace{\{\mathtt{X_2,X_3,X_4}\}}E_1\]
where, ${E_1}=\{\mathtt{type~X_2}=\lb{Prod}[\mathtt{X_2},~\mathtt{Empty},\mathtt{X_4}]~|~\mathtt{X_3}$\\
\hfil$
,\mathtt{type~X_3}=\lb{Int}[\mathtt{Empty}]
,\mathtt{type~X_4}=\lb{Int}[\mathtt{Empty}]\}$.\\
Finally, applying \textsc{(T-Nt1)}, $\mathtt{X_1}$ and the global set $E$:
\begin{equation*}
\left\{\begin{array}{llll}
\mathtt{type~X_1}=\mathtt{X_2}, &\mathtt{type~X_2}=\lb{Prod}[\mathtt{X_2},~\mathtt{Empty},\mathtt{X_4}]~|~\mathtt{X_3}, \\
\mathtt{type~X_3}=\lb{Int}[\mathtt{Empty}], & \mathtt{type~X_4}=\lb{Int}[\mathtt{Empty}]
\end{array}\right\}
\end{equation*}
are derived.

\section{Properties of Typing Rules for CPEG}\label{meta}
In this section, we show type soundness and uniqueness of types property of the typing rules.

First, we have shown the \emph{uniqueness of types} property, which means that the inferred type is always unique for a given CPEG. Furthermore, this property says that the derivation tree is also deterministic.

\begin{theorem}[Uniqueness of Types]
	In a given typing context $\Gamma$, an expression $\Xe$ has uniquely one type with type variables all in the domain of a global set of types $E$.
	Moreover, there is just one derivation of this typing built from the typing rule that derives the typing relation.
\end{theorem}
\begin{proof}
	The proof goes by the structural induction on $\Xe$.
\end{proof}

Next, we have shown type soundness.
Intuitively, type soundness states that if a CPEG has type $\mathtt{T}$ under the type inference, then the trees derived from the CPEG are typed by the type $\mathtt{T}$.

\begin{theorem}[Soundness]
	Let $G$ be a CPEG such that $G=(N_G,\Sigma,$
	$P_G,
	e_s,\mathcal{S})$.
	Let $E$ be a global set of regular expression types.
	\[\forall v\in\{v|\exists x,y.e_s\Downarrow_y^x v\}.~ \Gamma\vdash e_s:\mathtt{T}\valspace{\chi}E\Rightarrow v:\mathtt{T}\]
\end{theorem}
\begin{proof}
	By induction on derivation of $\Xe\Downarrow^x_y v$.
	The (\textsc{E-Empty}), (\textsc{E-Term1}), (\textsc{E-Not2}) cases are immediate by the rules (\textsc{T-Empty}), (\textsc{T-Term}), (\textsc{T-Not}), (\textsc{S-Empty}).
	For the other cases, we will discuss the (\textsc{E-Capture1}) case and the proof for the remaining rules proceed in the same manner.
	
	\textit{Case }\textsc{E-Capture1}: 
	The start expression is $\CAP{L}{\Xe}$ and the derived tree is $\lb{L}[v]$ such that $\Xe\Downarrow_y^x v$.
	Let the following be an induction hypothesis.
	$$\Gamma\vdash\Xe:\mathtt{T}\valspace{\chi}E\Rightarrow v:\mathtt{T}$$
	The only typing rule for $\CAP{L}{\Xe}$ is the rule (\textsc{T-Capture}).
	By inversion of (\textsc{T-Capture}), $$\Gamma\vdash\Xe:\mathtt{T}\valspace{\chi}{E}$$
	Using the induction hypothesis, $v$ has type $\mathtt{T}$. Finally, $\lb{L}[v]:\lb{L}[\mathtt{T}]$ by (\textsc{S-Node}).
\end{proof}
\section{Related Work}\label{relwork}

Here we review parsing with an emphasis on tree construction technique.

PEGs have been gaining popularity among many language developers since it was presented by B. Ford in 2004. In the context of PEGs, a parser generator is a standard approach to the development of parsers \cite{Grimm:2006:BET:1133981.1133987}.
PEGs, as well as other CFG-based grammars, can well describe the syntactic pattern of the input while providing a very poor specification of the output.
Consequently, the construction of syntax trees is mostly left to the embedded action code, which is usually written in a target programming language.
However, the use of action code reduces a good property of declarative grammar specification.
Moreover, their additional code generation and compilation process are cumbersome compared to regex-style matching and capturing.

Parser combinators such as Parsec \cite{leijen2002parsec}, FParsec \cite{tolksdorf2013fparsec}, and Scala Parser Combinators \cite{Moors08parsercombinators} provide a more integrated means for writing a recursive descent parser.
Although parser combinators are not always based on some grammar formalism, many PEG parsers have been implemented by the combinators.
In parser combinators, the resulting syntax trees can be well-typed.
The tree construction, however, relies on code fragments that directly manipulates parsing results. Consequently, the specification of a parser is still hard to maintain \cite{Klint:2010:IDT:1868281.1868291,Adams:2013:PPI:2429069.2429129}.

More recently, declarative parsing (no action code) has been focused in many grammar formalisms, since action code makes it difficult to maintain a parser specification and reduces grammar reusability \cite{Klint:2010:IDT:1868281.1868291,Adams:2013:PPI:2429069.2429129}.
In the contexts of PEGs, LPeg \cite{LPeg} is implemented as a PEG-based pattern matching tool that provides a grouped capture like regex.
Moreover, Nez grammar \cite{Nez,NezGrammar}, an ancestor of CPEG, provides a structured capture that can construct complex syntax trees. These parse data require no action code, but they are untyped.

For the parser users, types are significant. This viewpoint has made another attempt to parser generations from data types.
Notably, PADS/ML \cite{Mandelbaum:2007:PFD:1190216.1190231} describes DDC$^\alpha$ \cite{Mandelbaum:2007:PFD:1190216.1190231,Fisher:2006:NDD:1111320.1111039} and then generates a parser from the data specification. The ``parser from types'' approach can follow various syntax patterns including programming languages.
However, formally specifying the provided types safety would be more challenging \cite{Petricek:2016:TDM:2908080.2908115}.

Finally, a type system for grammar is new.
As a starting point, we use RETs, which have intensively been studied in the context of XML schemas and tree automata \cite{Hosoya:2003:XST:767193.767195,Hosoya:2005}.
We consider that  RETs are a straightforward type representation of syntax trees, and RETs can make a theoretical bridge between declarative parsing, tree automata, and programming language design. Indeed, binding RETs with ML and OCaml has been reported in \cite{Sulzmann:2006:TEX:1706640.1706940,Frisch:2006:OX:1159803.1159829}.

\section{Conclusion}\label{conc}
Regular expressions, or regexes, have had great success both as pattern matching and as a library tool to develop small parsers. However, the absence of recursive patterns results in very limited parser applications. Since PEGs are more powerful than regular expressions, PEGs with regex-like captures could make it much easier to integrate a full-fledged parser into programs.

CPEGs are a formally developed extension of PEGs with regex-like captures. Two annotations (capture and fold-capture) allow a flexible construction of complex syntax trees.
More importantly, a CPEG tree is a company with regular expression types that are a foundation of XML schema and tree automata.
A regular expression type for a given CPEG is inferred syntactically.
We present a formal definition of the type inference and we proved its soundness and uniqueness of types property.

Our attempt to type system for grammar is new.
There are several interesting issues that remain unexplored.
The future direction is that we will implement the CPEG based parser and investigate a practical aspect of CPEG. We intend to implement the CPEG based parser in F\# and integrate our type inference into F\# type provider \cite{syme2012strongly}.

%% Bibliography
\bibliographystyle{ACM-Reference-Format}
\bibliography{parser,mypaper,url,data,local,schemas}

\end{document}